\documentclass[sigconf, nonacm]{acmart}
\AtBeginDocument{%
  }

\renewcommand\footnotetextcopyrightpermission[1]{}
\begin{document}

\newcommand{\pending}[1]{%
  \textbf{\textcolor{red}{[PENDING: #1]}}}

\title{Same Book, Different Fills:
Partial Identification of FIFO Execution from Aggregate Order Books}



\author{Riya Danait}
\authornote{Corresponding author.}
\orcid{0009-0002-9306-470X}
\email{riya.danait@maths.ox.ac.uk}
\affiliation{%
  \institution{Mathematical Institute, University of Oxford}
  \city{Oxford}
  \state{}
  \country{UK}
}

\author{Yuliana Zamora}
\orcid{0000-0002-9133-3462}
\email{yzamora@nvidia.com}
\affiliation{%
  \institution{NVIDIA}
  \city{Santa Clara}
  \state{CA}
  \country{USA}
}

\author{Ioana Boier}
\orcid{0009-0007-9321-8218}
\email{iboier@nvidia.com}
\affiliation{%
  \institution{NVIDIA}
  \city{Santa Clara}
  \state{CA}
  \country{USA}
}

\renewcommand{\shortauthors}{Danait et al.}

\begin{abstract}
  Price-level limit-order-book (L2) data reveal aggregate liquidity but not the ordered queue of resting orders required by price--time-priority matching. Consequently, passive-execution backtests based on L2 can depend on an unobserved cancellation-allocation rule even when the observed prices, quantities, and trades are held fixed. 
  
  We formulate recovery of market-by-order histories from aggregate snapshots as a conditional partial-identification problem: multiple histories can reproduce the same aggregate path. Holding that path, reconciled market removals, latent order partitions, and additions fixed, our path-preserving compiler varies only cancellation allocation among front, quantity-weighted-random, and back rules. Within this compiler class, we establish front--back fill ordering for a virtual tagged order during a single touch-price spell.
  
  Our empirical study uses seven months of synchronized 2025 Tokyo Stock Exchange data for two instruments with different trading activity: RIC 1301.T and RIC 7911.T. Ten-level L2 snapshots provide aggregate book states, while L1 trades allow market removals and resting side to be inferred through reconciliation. Each instrument contributes 1,080 matched five-minute episodes across the same 18 held-out trading days. The aggressive benchmark is unchanged across FIFO realizations, but passive execution is sensitive to the cancellation rule. For 1301.T, front rather than back cancellation raises completion before terminal crossing by 8.01 percentage points and reduces implementation shortfall by 1.010 bps. For 7911.T, front cancellation raises completion by 7.39 percentage points and reduces shortfall by 0.384 bps.
  
  These findings show that observationally equivalent aggregate-book paths can imply economically different passive-execution outcomes. Execution policies evaluated from aggregate data should be accompanied by FIFO sensitivity analysis rather than reported as single-point estimates based on an unobservable queue assumption.
\end{abstract}

\ccsdesc[500]{Applied computing~Economics}
\ccsdesc[500]{Computing methodologies~Uncertainty quantification}
\ccsdesc[300]{Computing methodologies~Model verification and validation}

\keywords{limit order book, price--time priority, queue position, execution simulation, market microstructure, partial identification}


\maketitle

\section{Introduction}

Execution simulators commonly separate two tasks: a statistical market model evolves the displayed limit order book, and a matching engine applies the exchange's price-time-priority rules \cite{huang2015queue,frey2023jaxlob,jpxTransactionMethods}. When the market model or historical input is restricted to level two (L2), this division creates a gap. L2 reports total quantity at each price but not the constituent orders, their arrival order, or which order a cancellation removes. Those missing details leave displayed depth unchanged, but they determine when a passive order reaches the front of its queue and can be filled.

Existing L2-based execution studies address this gap by assuming a particular within-queue cancellation convention; for example, that cancellations occur near the front, near the back, or at sampled queue positions \cite{dixon2018execution,vyetrenko2019risk,coletta2022world}. In effect, they fix one order-level history
consistent with the observed aggregate path and then evaluate policies under that choice. We instead ask whether an execution conclusion is preserved across several FIFO-consistent order-level histories that reproduce the same observed L2 path. This is a conditional partial-identification problem: the observations restrict, but do not select, the compatible order-level history \cite{manski2003partial,tamer2010partial}. We call the resulting range a \emph{FIFO-sensitivity envelope}. It is conditional on a specified reconstruction procedure and is not an estimate of the unobserved historical queue. For example, suppose a passive child joins behind 300 shares at the best bid and later orders join behind it. A subsequent 100-share cancellation may remove older quantity ahead of the child or newer quantity behind it. Both cases produce the same L2 decrease, but only the former improves the child's queue position and fill probability.

We study seven months of synchronized 2025 Tokyo Stock Exchange data for two instruments. Ten-rank L2 snapshots provide aggregate book states, while a synchronized level-one (L1) trade feed provides trade time, price, and quantity. We reconcile the feeds, hold the resulting aggregate path fixed, and compile matched histories using front, quantity-weighted-random, and back cancellation allocation. A virtual order then measures passive fills without changing the replayed market. The same interface evaluates three price-taking execution policies. The reconstructed background streams are checked independently through JAX-LOB \cite{frey2023jaxlob}, a GPU-capable matching engine implemented in JAX, so that the reconstruction is verified in an external replay system rather than only by the compiler's internal accounting.

The paper makes three contributions. First, a precise definition of FIFO-consistent order-level histories compatible with an observed aggregate path; second, a reconstruction procedure that varies only within-queue cancellation placement and is checked independently through JAX-LOB replay; and third, a held-out evaluation of passive-fill ranges, changes in the preferred policy, and the cost of choosing under one queue assumption and evaluating under another. To our knowledge, prior work has not held the aggregate path fixed, constructed multiple compatible FIFO histories, checked their replay in an independent matching engine, and then measured how execution costs and policy choices vary across them.

\section{Related Work}

\subsection{Queue state, fills, and execution}

Price--time priority gives earlier orders economic value \cite{donnelly2018timepriority,moallemi2017queue}, while empirical cancellations are not uniformly distributed through a queue \cite{gu2013cancellation}. Dixon represents unknown cancellation quantity ahead of a reference order through favorable and unfavorable cases \cite{dixon2018execution}. Vyetrenko and Xu compare front, back, and uniform allocation in L2 replay \cite{vyetrenko2019risk}; World Agent models cancellation depth and within-queue position as part of a learned simulator \cite{coletta2022world}. Maglaras, Moallemi, and Wang estimate conditional time-to-fill distributions directly from observed market states \cite{maglaras2022fill}. That predictive problem takes the training fill outcome as observed, whereas our problem concerns the range of fills that remains after L2 aggregation has removed the order identities needed to construct a unique FIFO outcome. We retain the complete aggregate path and all common inputs, and vary only the allocation of residual within-queue cancellations. Yueshen studies queue uncertainty induced by latency and venue fragmentation \cite{yueshen2025queuing}; the uncertainty here arises ex post because aggregation has discarded historical order identities.

\subsection{Queue-reactive and learned market models}

Queue-reactive models make limit, cancellation, and market-event intensities functions of displayed queues \cite{huang2015queue}. This approach builds on earlier Markovian queuing models of order book dynamics \cite{cont2013}. GAN-based dynamic calibration also allows order-flow simulations to capture intraday patterns and changing volatility conditions \cite{prenzel2022}. Later work adds event sizes, signals, cross-level dependence, and neural state representations \cite{bodor2024sizes,bodor2025mdqr,sfendourakis2026multidimensional}, including environments for execution-policy learning \cite{espana2025rl}. These models generate aggregate quantities, but aggregate accuracy alone does not determine the persistent order identities needed for FIFO fills. The present study addresses this downstream order-level realization problem.

\subsection{Simulation and validation}

ABIDES supports asynchronous agent-based simulation \cite{byrd2020abides}; JAX-LOB provides parallel price--time-priority matching \cite{frey2023jaxlob}. Realism tests compare market statistics and stylized facts \cite{vyetrenko2020getreal,nagy2025lobbench}. Our test is complementary: given one observed aggregate path, it asks whether alternative mechanically valid order-level refinements produce the same execution decision.

\section{Path-Conditional FIFO Sensitivity Framework}

\subsection{Observed aggregate path}

Let \(A_b=\{(t_k,Q_k,T_k,Z_k)\}_{k=0}^{K_b}\) be one retained contiguous replay block. Here \(Q_{s,p,k}\) is displayed quantity on side \(s\in\{b,a\}\) at price \(p\) after update \(t_k\); \(T_k\) contains the regular L1 trades assigned to that transition; and \(Z_k\) records session, outer-rank, and unresolved-transition flags. A side--price transition is retained only when the price is observed and its change can be reconciled without an unresolved compound update. Let \(\mathcal I(A_b)\) denote these retained cells.

For a retained cell, displayed quantity satisfies
\begin{equation}
  \Delta Q_{s,p,k}
  =
  L_{s,p,k}-D_{s,p,k}-M_{s,p,k},
  \qquad
  \Delta Q_{s,p,k}=Q_{s,p,k}-Q_{s,p,k-1},
  \label{eq:accounting}
\end{equation}
where \(M\) is an L1-attributed market removal. L1 supplies time, price, and quantity, but not the resting side. We assign a trade to a side--price cell only when its price is uniquely visible on one side of the pre-event book and the L2 reduction supports its quantity. After this assignment, define
\begin{equation}
  r_{s,p,k}=\Delta Q_{s,p,k}+M_{s,p,k},
  \qquad
  L_{s,p,k}=[r_{s,p,k}]_+,\quad
  D_{s,p,k}=[-r_{s,p,k}]_+ .
  \label{eq:residual}
\end{equation}
\([x]_+=\max(x,0)\). This is a minimal one-sided decomposition: after the attributed trade, an accepted cell may contain an addition or a cancellation, but not an unresolved simultaneous amount of both. Cells that do not support this interpretation end the replay block. Within each accepted transition, reconstructed actions are applied in the common order \(M\), then \(D\), then \(L\): reconciled market removals are processed first, followed by residual cancellations and then residual additions. This is a reconstruction convention shared by all FIFO rules.

An order-level realization \(R\) replaces each aggregate quantity with an ordered queue of order IDs, remaining quantities, and priority times. Its aggregate projection is
\begin{equation}
  \Gamma(R)_{s,p,k}
  =
  \sum_{(i,q,\tau)\in\mathcal O_{s,p,k}(R)}q .
  \label{eq:projection}
\end{equation}
A realization is compatible with block \(A_b\) when it reproduces the aggregate path and evolves through valid order-level updates. An addition of total size \(L_{s,p,k}\) appends live quantity at the corresponding side--price queue; an L1-attributed market removal of size \(M_{s,p,k}\) consumes resting quantity in FIFO order; and a cancellation of size \(D_{s,p,k}\) removes sufficient live background quantity. Each realization is restricted to the displayed finite-depth support; no liquidity is imputed outside the observed ranks. For the boundary-fixed view, valid evolution additionally permits the deterministic finite-rank synchronization operations encoded by \(Z_k\). These operations are identical across FIFO rules; strict blocks end before such an operation is required. Define
\begin{equation}
\mathcal R(A_b)=
\left\{
R\;\middle|\;
\begin{array}{l}
\Gamma(R)_{s,p,k}=Q_{s,p,k}
\quad\forall (s,p,k)\in\mathcal I(A_b),\\
R\text{ preserves }(t_k,T_k,Z_k)_{k=0}^{K_b},\\
R\text{ obeys the valid }(L,D,M)\text{ evolution and the common} \\
\text{boundary-synchronization rule described above}
\end{array}
\right\}.
\label{eq:compatible}
\end{equation}
Thus the observations determine a set of compatible queue histories rather than one recovered history.

\subsection{Controlled compiler class}

Here, a \emph{compiler} is simply a reconstruction procedure that converts an aggregate path into an order-level message stream. We evaluate a specified subset of \(\mathcal R(A_b)\). Let
\[
\mathcal G
=
\{\mathrm{front},\mathrm{random},\mathrm{back}\},
\]
where \(\mathrm{random}\) denotes the quantity-weighted-random rule defined below. Let \(\omega=(\omega_P,\omega_C)\) contain the partition seed and the separately declared random cancellation seed. For each \(g\in\mathcal G\) and \(\omega\in\Omega\), write \(R^{g,\omega}=\mathsf C_{g,\omega}(A_b)\). The compiler class is
\begin{equation}
\begin{aligned}
\mathcal R_{\mathfrak C}(A_b)
&=
\bigl\{
R^{g,\omega}:
g\in\mathcal G,\,
\omega\in\Omega
\bigr\} \\
&\subseteq
\mathcal R(A_b).
\end{aligned}
\label{eq:compilerclass}
\end{equation}

For fixed common inputs \(\omega\), define \(\mathcal R_{\mathfrak C}(A_b\mid\omega)
=\{R^{g,\omega}:g\in\mathcal G\}\). Then \(\mathcal R_{\mathfrak C}(A_b)
=\bigcup_{\omega\in\Omega}
\mathcal R_{\mathfrak C}(A_b\mid\omega)\). Within each fixed-\(\omega\) slice, \(\omega_P\) fixes identical initial and subsequent addition partitions; timestamps, prices, L1-attributed removals, and block boundaries are also common. The rules change only how a within-book residual cancellation is allocated: \emph{front} removes the oldest eligible quantity, \emph{back} removes the newest, and \emph{random} successively samples an eligible live background order with probability proportional to its remaining quantity, removes up to the outstanding cancellation amount, and repeats until \(D_{s,p,k}\) has been allocated, using \(\omega_C\).

The primary 300-second policy analysis uses a boundary-fixed view. It ends at unresolved transitions and session boundaries but admits finite-rank entry and exit through a deterministic synchronization rule shared by all three FIFO realizations. Boundary removals use the same back-of-queue convention under every rule and are not treated as ordinary within-queue cancellations governed by the front, random, or back rules. Tagged probes use strict blocks that additionally end before every outer-rank entry or exit. A stricter policy analysis is reported on the subset of strict blocks that can support a 300-second episode.

\subsection{Shadow policies and outcome sets}

Each realization is replayed as a canonical background book. Only background messages are sent to JAX-LOB. A passive policy order is represented separately by a virtual queue marker that records its price, remaining quantity, and background quantity ahead. Same-price additions after submission join behind the marker. A cancellation reduces quantity ahead only when the corresponding realization removes a background order ahead of the marker.

At a background execution of size \(m_k\) at the marker price, let \(H_{g,k^-}\) denote quantity ahead immediately before the execution and let \(x_{g,k^-}\) denote the marker's remaining quantity. Its virtual fill and post-execution state are
\begin{align}
f_{g,k}
&=
\min\!\left\{
[m_k-H_{g,k^-}]_+,\,
x_{g,k^-}
\right\},
\label{eq:virtual-fill}\\
H_{g,k}
&=
[H_{g,k^-}-m_k]_+,
\qquad
x_{g,k}
=
x_{g,k^-}-f_{g,k}.
\label{eq:virtual-state}
\end{align}
For initial quantity \(X\), cumulative fill is \(F_g(t)=X-x_g(t)\). The cap by \(x_{g,k^-}\) prevents overfilling and makes all subsequent fills zero once the marker is complete. The marker never diverts an execution or changes the background book.

An aggressive policy action is evaluated on a copy of the contemporaneous book. The copy is consumed to obtain executable quantity and price and then discarded. Consequently, neither passive nor aggressive actions affect the next historical transition. These are price-taking counterfactuals.

For policy \(\pi\), let \(Y(\pi;R)\) contain its prespecified fill and cost outcomes. Conditional on the event-reconstruction and finite-depth assumptions above, the aggregate path defines a compatible outcome set, whereas the experiment evaluates the subset generated by the specified compiler:
\begin{equation}
\begin{aligned}
\mathcal Y(\pi;A_b)
&=
\{Y(\pi;R):R\in\mathcal R(A_b)\},\\
\mathcal Y_{\mathfrak C}(\pi;A_b)
&=
\{Y(\pi;R):R\in\mathcal R_{\mathfrak C}(A_b)\}\\
&\subseteq
\mathcal Y(\pi;A_b).
\end{aligned}
\label{eq:outcomeset}
\end{equation}
The inclusion need not be an equality. Thus \(\mathcal Y_{\mathfrak C}\) is the sensitivity set produced by our stated reconstruction rules, not the complete set of outcomes allowed by the data. At fixed \(\omega\), front and back are endpoint conventions for cancellation location. This does not make them extrema of \(Y\) or \(\ell\) over the full compiler class. The only endpoint result we claim is the tagged-fill ordering below under its stated single-touch conditions; front and back are not sharp bounds over all histories compatible with \(A_b\).

\subsection{Conditional ordering for a passive order}

Consider a period during which one price remains the best quote on its side. A virtual order joins that price at time \(\tau\); later additions join behind it, executions consume FIFO, and background cancellations cannot cancel it.

\begin{proposition}[conditional tagged-fill ordering]\label{prop:cond-tagged-fill-ordering}
    Let \(r\) denote the quantity-weighted-random rule, and let \(c_{g,k}\) denote the amount of the \(k\)th cancellation allocated to background quantity ahead of the marker under \(g\in\{\mathrm{front},r,\mathrm{back}\}\). Suppose the realizations share the same initial queue, subsequent additions, and executions, and that, for every cancellation \(k\),
    \[
    c_{\mathrm{back},k}
    \leq c_{r,k}
    \leq c_{\mathrm{front},k}.
    \]
    Then
    \begin{equation}
      F_{\mathrm{back}}(t)
      \leq
      F_r(t)
      \leq
      F_{\mathrm{front}}(t)
      \label{eq:fillorder}
    \end{equation}
    for every \(t\) before the best price changes.
\end{proposition}

\textit{Proof sketch.}
At submission, the three markers have the same remaining quantity and the same quantity ahead. Subsequent additions join behind the markers and therefore do not change their queue-ahead quantities. At cancellation \(k\),
\[
H_{g,k}^{+}
=
[H_{g,k}^{-}-c_{g,k}]_+ .
\]
The assumed ordering of ahead-of-marker cancellation quantities therefore preserves
\[
H_{\mathrm{front},k}
\leq
H_{r,k}
\leq
H_{\mathrm{back},k}.
\]
At a common execution of size \(m_k\), the states update according to
\[
H_{g,k}^{+}
=
[H_{g,k}^{-}-m_k]_+,
\qquad
x_{g,k}^{+}
=
\max\!\left\{
x_{g,k}^{-}-[m_k-H_{g,k}^{-}]_+,
0
\right\}.
\]
These updates preserve
\[
x_{\mathrm{front},k}
\leq
x_{r,k}
\leq
x_{\mathrm{back},k}.
\]
Induction over the common event sequence gives this ordering at every subsequent event. Since \(F_g(t)=X-x_g(t)\), the stated fill ordering follows. \(\square\)

The proposition is conditional in two respects: it applies to a single passive marker during an unchanged best-price spell, and the intermediate rule must satisfy the ahead-removal ordering at each cancellation. We do not claim that every random allocation satisfies this premise or that the result extends to arbitrary repricing or multi-order policies. In the empirical analysis, we verify the resulting fill ordering for every matched tagged-probe triplet.



\section{Data and Event Reconstruction}

\subsection{Feeds and scope}

We use proprietary LSEG data for Tokyo Stock Exchange equities from January 6 through July 25, 2025 \cite{lsegtyodata}. The sample contains two instruments, RIC 1301.T and RIC 7911.T, with synchronized ten-level L2 snapshots and L1 trade records. Each normalized L2 row reports occupied price ranks and aggregate size, but not usable order counts. The L1 feed contains quotes, trades, and status records; regular trades provide exchange time, price, and volume. We exclude pre-open, lunch, auction, and other non-continuous phases.

We retain 1301.T as the lower-turnover instrument and use 7911.T as a more actively traded instrument. 7911.T was selected before examining June or July policy outcomes using regular-trade yen turnover on three fixed development dates: January 15, March 14, and May 15. Candidates were required to appear in both feeds and pass basic two-sided-book and ten-level-completeness checks. The resulting comparison tests whether FIFO sensitivity persists across different activity and tick-size regimes. It should not be interpreted as a causal estimate of liquidity, since the instruments also differ in price, tick size, and parent-order scale.

\subsection{Causal synchronization}

Within each feed, capture timestamp and internal sequence determine order. Exchange time is the cross-feed key, while capture time provides a capture-lag diagnostic. A January orientation test assigned same-key trades either to the transition ending at that key or to the following transition. The former explained 99.85\% of regular trade quantity by a same-price reduction, versus 8.66\% for the latter, and was fixed before later months were processed.

We consider a date to be \textit{feed-eligible} only when both feeds contain the RIC, scans are complete, at least 98\% of regular trades fall within L2 coverage and have an exact exchange-time match, the daily 99th-percentile absolute capture lag is below 20 ms, and at least 99\% of continuous rows are two-sided and uncrossed with all requested ten ranks present in at least 95\% of rows. 

\subsection{Trade reconciliation}

Same-key trades are grouped by price. Supported quantity is assigned when the trade price is uniquely visible on one resting side; Equation~\eqref{eq:residual} then classifies the remaining one-sided change. The same procedure is applied independently to both instruments.

The primary FIFO study starts from all 136 and 137 feed-eligible dates for 1301.T and 7911.T, respectively. Valid contiguous price-keyed blocks are constructed separately from each instrument's complete reconciliation ledger. This observed-path view preserves the available aggregate book path and reconciled market removals directly. Table~\ref{tab:data} summarizes the resulting feed-eligible samples and reconciliation diagnostics.

\begin{table}[t]
\caption{TYO data and reconciliation, January--July 2025.}
\label{tab:data}
\centering
\setlength{\tabcolsep}{3pt}
\begin{tabular}{@{}lrr@{}}
\toprule
Metric & 1301.T & 7911.T \\
\midrule
Tick size (JPY) & 5 & 1 \\
Feed-eligible dates & 136 & 137 \\
Continuous L2 rows & 843,135 & 9,062,080 \\
Regular L1 trades & 24,607 & 319,046 \\
Exact-time-match median & 100\% & 100\% \\
Monthly p99 lag (ms) & 3.39--4.57 & 4.43--4.89 \\
Same-price trade support & 99.85--99.99\% & 99.77--99.99\% \\
Unresolved transitions & 4.14--6.11\% & 5.09--6.59\% \\
Boundary transitions & 11.10--16.00\% & 8.39--12.32\% \\
Held-out dates / episodes & 18 / 1,080 & 18 / 1,080 \\
\bottomrule
\end{tabular}
\end{table}

\subsection{Chronological protocol}

For 1301.T, January--May contain 97 feed-eligible dates used to estimate the partition inputs and parent-order scale; June contains 21 dates used to select block, probe, and policy settings; and July contains 18 held-out dates. Estimation and selection are performed separately for each instrument. Instrument selection, compiler rules, exclusions, metrics, and robustness checks are fixed before any July FIFO outcomes are examined. No events are randomly split or pooled across periods; the reported FIFO effects therefore describe the July held-out distribution, not the January--June development and selection periods.

JPX classifies security code 1301 as a TOPIX Small 2 constituent and 7911 as a TOPIX Mid400 constituent during the sample period \cite{jpxTopix2024}. Security 1301 therefore follows the ``Other Issues'' tick schedule; its retained prices lie in the band for which the minimum increment is \(5\) yen. Security 7911 is a TOPIX500 issue and its retained prices lie in the corresponding \(1\)-yen tick band \cite{jpxTickSize}. TSE domestic stocks trade in units of 100 shares \cite{jpxTradingUnit}. The protocol therefore uses tick sizes of \(5\) yen for 1301.T and \(1\) yen for 7911.T, with a 100-share exchange trading unit for both.

\section{Path-Preserving FIFO Compiler}

\subsection{Common queue construction}

At the first snapshot of each block, every price-level quantity is partitioned into latent background orders. Because L2 does not report those orders, the partition mechanism uses positive reconciled aggregate addition increments from January--May as drawing inputs. Such increments may combine several historical submissions and are not claimed to be individual-order sizes. Draws are truncated so that their sum equals the displayed quantity. The same seeded partitions are used under all cancellation rules.

Subsequent positive increments are partitioned by the same mechanism and appended to the queue. L1-attributed market removals consume FIFO quantity; residual cancellations follow the selected front, random, or back rule. Every rule removes the same total quantity and reproduces the next observed aggregate state.

\subsection{Executable invariants}

A block is accepted only if times are monotone, live order IDs are unique, cancellations reference sufficient live quantity, executions consume eligible resting orders in FIFO order, and no quantity is negative. After every retained transition, summing live orders at each price must recover the observed L2 state. The three cancellation streams must also have identical aggregate states, additions, L1-attributed removals, times, and block boundaries.

The aggregate projection is checked after every accepted transition. As a separate check, a specified sample of reconstructed background streams is replayed through JAX-LOB and compared with the compiler book at the same points in time. This independent replay checks message routing and matching in an external JAX-based engine rather than only through the compiler's own state updates. Virtual policy orders remain outside JAX-LOB and are evaluated through their queue-ahead records.

\subsection{Computational setup}

Experiments were conducted on a single node with 8 CPU cores and 4 NVIDIA GB200 GPUs. Final matching-engine verification used the JAX GPU backend in an NVIDIA PyTorch 25.01 NGC container. The primary held-out design contained 6,480 matched episode–side–seed units per ordered FIFO-rule comparison, with each unit evaluated under the three cancellation rules and policies. 

\section{Experimental Design}

\subsection{Tagged queue-entry probes}

The primary diagnostic registers a virtual 100-share marker at the same-side best price every five seconds when both sides are present and the spread is positive. This is a virtual sensitivity probe, not a submitted child order. A probe ends after 15 seconds, when that price is no longer the best price, or when its replay block ends. Buy and sell probes are paired across front, random, and back rules with the same path, submission time, size, partition, and random seed. Outcomes are fill fraction, any fill, time to 50\% and 100\% fill, and initial quantity ahead. Equation~\eqref{eq:fillorder} is checked for every eligible matched triplet.

\subsection{Price-taking execution policies}

Within each eligible block, 300-second episodes are placed consecutively from the earliest valid start time without overlap. A terminal block fragment shorter than 300 seconds is discarded. Episode boundaries are identical across policies, sides, seeds, and FIFO rules. We evaluate three deterministic shadow policies on both sides:
\begin{description}
  \item[Aggressive TWAP] follows cumulative targets at ten equally spaced checkpoints; targets are rounded half up to 100-share lots and zero-quantity children are skipped.
  \item[Passive at touch] joins the best same-side price and retains its queue marker while that price remains best. After a completed background transition changes the best price, the marker is canceled and the policy rejoins the new best price after that transition.
  \item[Hybrid] follows the passive rule but crosses any schedule deficit at the ten checkpoints.
\end{description}

Passive children use virtual queue markers; aggressive children query copied book states. Parent size is 5\% of development-period median positive same-side aggressive volume, rounded to 100-share lots with a one-lot minimum: 100 shares for 1301.T and 200 for 7911.T. Background transitions precede policy actions at tied timestamps.

\subsection{Metrics}

Let \(X\) be parent quantity, \(P_0\) arrival midprice, and \(\{(q_i,P_i)\}\) the natural fills before horizon \(T\). The residual quantity is
\[
q_T=X-\sum_i q_i .
\]
When \(q_T>0\), it is crossed against a copy of the terminal book and \(P_T\) denotes its volume-weighted execution price; when \(q_T=0\), the terminal contribution below is defined to be zero. We define the implementation shortfall 
\begin{equation}
  \mathrm{IS}(\pi;R)
  =
  10^4\sigma
  \frac{\sum_iq_i(P_i-P_0)+q_T(P_T-P_0)}{XP_0}.
  \label{eq:is}
\end{equation}
Here \(\sigma=1\) for buys and \(-1\) for sells, so lower values of \(\mathrm{IS}\) indicate better execution. This completed implementation shortfall is the scalar policy loss \(\ell(\pi;R)\). We separately report pre-terminal completion \(\sum_iq_i/X\), the passive fraction \(X^{-1}\sum_{i\in\mathcal P}q_i\), where \(\mathcal P\) indexes fills obtained through virtual passive markers, the terminal-cross fraction, and natural completion time \cite{perold1988shortfall,almgren2001optimal}. An episode without enough displayed terminal depth to execute \(q_T\) is flagged under the specified eligibility rule rather than assigned an imputed price.

For fixed common inputs \(\omega\), a policy selected under rule \(g\) has regret under rule \(h\)
\begin{equation}
 \operatorname{Regret}_{g\rightarrow h,\omega}
 =
 \ell(\widehat\pi_{g,\omega};R^{h,\omega})
 -
 \min_{\pi}\ell(\pi;R^{h,\omega}),
 \label{eq:regret}
\end{equation}
where \(\widehat\pi_{g,\omega}\) minimizes loss under \(R^{g,\omega}\). Policies whose losses differ by at most \(10^{-6}\) bps are treated as tied; a fixed alphabetical rule makes the choice deterministic. We record a policy-choice change only when the policy chosen under rule \(g\) is more costly than the best policy under rule \(h\) by more than this tolerance. We report front--back outcome differences, how often the policy choice changes, and the additional cost caused by such a change.

\subsection{Analysis units and uncertainty}

Unresolved transitions and session boundaries split both replay views. Tagged probes use strict blocks that also split at every outer-rank entry or exit, whereas the primary policy analysis uses the shared boundary-fixed convention. Within each instrument, probe and policy outcomes are first averaged within each held-out date, and the reported overall estimates weight the 18 dates equally. For each fixed \(\omega\), matched rule outcomes are differenced before date aggregation. The reported estimates average these paired contrasts over the prespecified seeds and dates; they are not suprema or envelope widths over the full compiler class. Each resampling replicate draws 18 complete dates with replacement and computes their equally weighted mean; the reported 95\% interval is the 2.5th--97.5th percentile range over 5,000 replicates. The frequency with which the policy choice changes, and the associated extra cost, are pooled over matched episode--side--seed observations and reported descriptively rather than with date-resampling intervals. Robustness checks compare the empirical-increment partition with deterministic 100-share chunks and compare boundary-fixed episodes with stricter block eligibility.

\section{Results}\label{sec:results}

\subsection{FIFO compilation and queue-entry sensitivity}

For each instrument, the held-out policy evaluation contains 1,080 distinct 300-second episode windows from the same 18 dates in July. For 1301.T, their union covers 97.27\% of continuous-session elapsed time, 95.27\% of L2 transitions, and 96.08\% of reconciled trade volume; the corresponding coverage for 7911.T is 96.80\%, 96.06\%, and 94.40\%. Each episode is evaluated for two sides and three primary seeds, yielding 6,480 matched episode--side--seed units per instrument and directed FIFO-rule comparison.

As a separate matching-engine check, 648 reconstructed streams from 54 background paths for 1301.T and 456 streams from 38 paths for 7911.T were replayed through JAX-LOB. All requested market-message quantities executed: 232,800 shares for 1301.T and 4,119,600 for 7911.T; no add message triggered an immediate trade, and no sentinel-price artifact remained. This shows that the reconstructed streams are not only aggregate-consistent but also mechanically executable in an independent JAX-based matching engine. The tagged-order diagnostic comprises 3,209 blocks for 1301.T and 21,309 for 7911.T. These cover 88.36\% and 87.53\% of elapsed time, 61.19\% and 48.72\% of transitions, but only 22.84\% and 18.12\% of reconciled trade volume, respectively. Mean 15-second fill is approximately 0.01\% for 1301.T; for 7911.T it ranges from 0.22\% under back removal to 0.50\% under front removal. These short-horizon probes test the predicted fill ordering rather than policy-level economic effects, which are measured using the separate matched 300-second policy episodes. All 354,774 and 370,992 matched triplets, respectively, satisfy the conditional ordering in Equation~\eqref{eq:fillorder}.

\begin{table*}[t]
\caption{Held-out policy outcomes by instrument. Lower implementation
shortfall (IS) is better. Bracketed values are 95\% equal-date
resampling intervals over 18 held-out dates. ``Random'' denotes
quantity-weighted-random cancellation allocation.}
\label{tab:fifo-results}
\centering
\scriptsize
\begin{tabular}{@{}llccc@{}}
\toprule
RIC (parent) & Policy and metric & Front & Random & Back \\
\midrule
1301.T (100)
 & Aggressive benchmark IS (bps)
 & \(10.010\,[9.701,10.316]\)
 & \(10.010\,[9.701,10.316]\)
 & \(10.010\,[9.701,10.316]\) \\
 & Hybrid IS (bps)
 & \(9.277\,[8.905,9.665]\)
 & \(9.469\,[9.138,9.813]\)
 & \(9.675\,[9.378,9.970]\) \\
 & Passive-at-touch IS (bps)
 & \(8.008\,[7.480,8.531]\)
 & \(8.418\,[7.941,8.890]\)
 & \(9.018\,[8.656,9.401]\) \\
\addlinespace
7911.T (200)
 & Aggressive benchmark IS (bps)
 & \(2.400\,[2.334,2.462]\)
 & \(2.400\,[2.334,2.462]\)
 & \(2.400\,[2.334,2.462]\) \\
 & Hybrid IS (bps)
 & \(1.039\,[0.911,1.160]\)
 & \(1.219\,[1.102,1.331]\)
 & \(1.453\,[1.332,1.579]\) \\
 & Passive-at-touch IS (bps)
 & \(0.765\,[0.547,0.974]\)
 & \(0.926\,[0.703,1.126]\)
 & \(1.149\,[0.932,1.357]\) \\
\addlinespace
1301.T (100) & Passive preterminal completion
 & 18.15\% & 15.29\% & 10.14\% \\
7911.T (200) & Passive preterminal completion
 & 97.20\% & 95.52\% & 89.81\% \\
\bottomrule
\end{tabular}
\end{table*}

\begin{figure*}[t]
  \centering
  \includegraphics[width=0.82\textwidth]
    {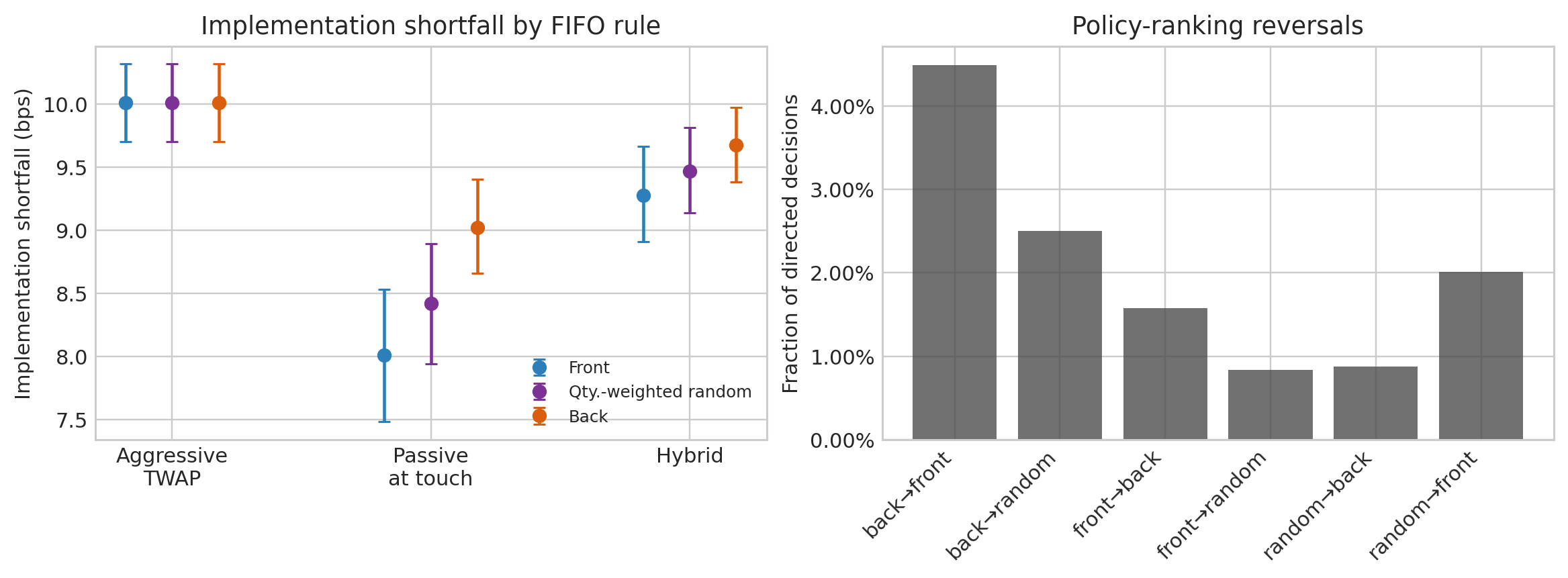}
  \includegraphics[width=0.82\textwidth]
    {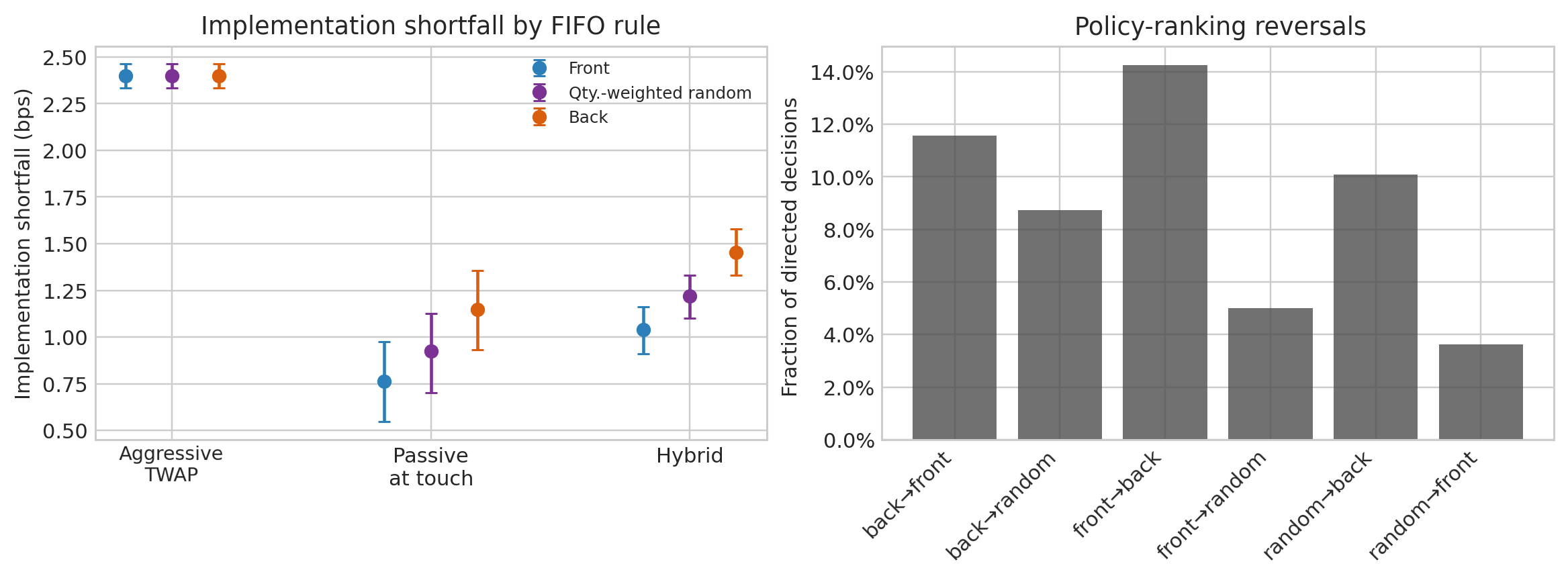}
  \caption{Held-out execution sensitivity across instruments. Top: 1301.T with a 100-share parent. Bottom: 7911.T with a 200-share parent. Left panels report mean implementation shortfall under front, quantity-weighted-random, and back cancellation allocation; error bars denote 95\% equal-date resampling intervals. Right panels report the fraction of directed decisions in which the lowest-cost policy selected under one queue rule is no longer lowest-cost under another.}
  \Description{For both instruments, the aggressive benchmark is unchanged across cancellation rules. Passive and hybrid implementation shortfall are lowest under front cancellation and highest under back cancellation. Policy-choice changes are more frequent for 7911.T than for 1301.T.}
  \label{fig:policy-sensitivity}
\end{figure*}

\subsection{Policy consequences}

Table~\ref{tab:fifo-results} and Figure~\ref{fig:policy-sensitivity} summarize the held-out policy outcomes. For 1301.T, the aggressive benchmark is invariant at 10.010 bps. Front rather than back removal raises passive preterminal completion by 8.01 percentage points and gives a front-minus-back shortfall difference of \(-1.010\) bps (95\% interval \([-1.253,-0.788]\)); the hybrid difference is \(-0.398\) bps. For 7911.T, the aggressive benchmark is invariant at 2.400 bps. Front rather than back removal raises passive preterminal completion by 7.39 percentage points and gives a front-minus-back shortfall difference of \(-0.384\) bps (\([-0.493,-0.284]\)); the hybrid difference is \(-0.414\) bps.

Passive-at-touch has the lowest date-average cost for both instruments under every cancellation rule, followed by hybrid and the aggressive benchmark. For 1301.T, the lowest-cost policy changes in 0.83\%--4.49\% of directed cross-rule comparisons. The corresponding range for 7911.T is 3.63\%--14.26\%. Conditional on a change, the added cost is 14.277--17.354 bps for 1301.T but 2.366--3.723 bps for 7911.T. Thus, policy choices change more frequently for 7911.T but at a smaller conditional cost. Because observations share calendar dates and background paths, these frequencies and costs describe the evaluated matched samples; they are not estimates of an unconditional market-wide frequency.

Deterministic 100-share partitions leave the primary front-minus-back estimates unchanged for both instruments. The stricter 1301.T replay-block analysis produces a front-minus-back passive-shortfall difference of \(-0.397\) bps, with interval \([-0.576,-0.229]\). The corresponding 7911.T strict-block comparison gives \(-0.325\) bps, with interval \([-0.649,0.000]\). Its direction is consistent with the primary result, but the small strict-block sample is not independently conclusive. Moreover, the boundary-fixed and strict analyses retain different episode samples. Their difference therefore shows sensitivity to episode eligibility; it does not isolate the effect of boundary handling on an identical sample.

\section{Discussion}

\subsection{What a qualified fill enables}

The sensitivity envelope does not recover the historical fill that would have occurred under the unobserved order-level book. Instead, it separates conclusions supported by the observed aggregate path from conclusions that depend on a queue-reconstruction assumption. A policy preference is robust within the declared compiler when its loss remains lower across all matched front, random, and back realizations. When the selected policy changes, the L2 record alone does not identify a unique decision, even though every realization reproduces the same displayed book and market removals.

This distinction supports a policy-acceptance test. A practitioner can require a candidate policy to remain preferred throughout the realization set, evaluate it using worst-case loss or minimax regret \cite{manski2007decision}, or impose a materiality threshold below which apparent improvements are treated as economically indistinguishable. If losses are evaluated over  the full rule-and-seed compiler class, one possible ambiguity-averse decision is

\begin{equation}
  \pi_{\mathrm{rob}}
  \in
  \arg\min_{\pi\in\Pi}
  \sup_{R\in\mathcal R_{\mathfrak C}(A_b)}
  \ell(\pi;R).
  \label{eq:robust-policy}
\end{equation}

This is a prospective full-class criterion. Section~\ref{sec:results} reports matched fixed-\(\omega\) rule contrasts and does not estimate the supremum in Equation~\eqref{eq:robust-policy}. Within the declared compiler class, a wide envelope can motivate a less passive policy or justify the cost of obtaining MBO data; a narrow envelope shows only that the decision is insensitive to the cancellation rules considered here. 

For learned execution policies, the same procedure could serve as a validation or robust-training layer. Training or testing across matched FIFO realizations could reduce the risk that an agent appears successful because it has exploited one simulator-specific fill convention. The resulting claim is deliberately conditional: robustness holds with respect to the declared compiler class. The purpose of the three policies in this study is to test whether a fixed policy comparison survives uncertainty about the hidden FIFO history.

\subsection{Interpretation of the endpoints}

Front and back are endpoint conventions for cancellation location, not estimates of the historical queue. They produce ordered endpoints for the single-touch tagged-fill functional under Proposition~\ref{prop:cond-tagged-fill-ordering}, but are not guaranteed extrema of general policy outcomes within \(\mathcal R_{\mathfrak C}(A_b)\), nor sharp bounds over \(\mathcal R(A_b)\). The quantity-weighted-random rule is also a modeling choice. Because paths, partitions, additions, L1-attributed removals, and block boundaries are paired, a matched comparison isolates the effect of the stated within-queue cancellation convention, conditional on the common partition and boundary rules.

The empirical purpose is to show, on a held-out path with broad within-instrument coverage, that observationally equivalent aggregate histories can support different costs and policy choices. Unlike a comparison of separately simulated trajectories, each comparison here holds aggregate states, reconciled market removals, partitions, additions, episode boundaries, and terminal execution rules fixed. The measured range therefore reflects reconstruction sensitivity conditional on one path, rather than a mixture of queue uncertainty and market-path forecast error.

The cross-instrument replication shows that the direction of FIFO sensitivity is not specific to the primary security. Front removal is most favorable to passive execution, random removal lies between the endpoints, and back removal is least favorable for both instruments. The economic magnitude differs, however. Passive completion is low for 1301.T, so the cancellation rule changes how much inventory must be crossed at the horizon and produces a relatively wide cost range. Most 7911.T passive inventory fills under every rule, so a similar completion difference translates into a narrower cost range. This comparison is descriptive rather than causal: parent quantities differ, as do tick size, activity, depth, and other instrument characteristics.

\subsection{Future extensions}

Further work should extend the design to a broader cross-section of securities, additional held-out months, parent sizes, and execution horizons. A further extension is scaling up the experiments. Since the reconstruction and replay interface are expressed in terms of order-level message streams and external matching-engine validation, the same framework can be applied to larger collections of synthetic or historical aggregate paths without changing the identification logic. JAX-based replay makes it natural to evaluate many matched FIFO realizations in parallel and to estimate sensitivity envelopes, worst-case loss, and cross-realization regret over broader scenario sets than in the present two-instrument study. Another extension is robust learning. Instead of using matched FIFO realizations only for post-hoc evaluation, one could train or test execution agents across multiple compatible queue realizations of the same aggregate path, so that good performance does not depend on a single simulator-specific fill convention.

\section{Limitations}

The study covers two RICs and seven months, but held-out evidence comes from one month and the same 18 calendar dates for each instrument, giving 36 instrument-day observations rather than a broad market cross-section. L1/L2 reconciliation attributes market removals without observing resting order IDs, while the latent partition and three cancellation rules remain modeling choices; front and back are conditional endpoints. Parent quantities differ across instruments---100 shares for 1301.T and 200 shares for 7911.T---so differences in absolute cost or completion cannot be interpreted as the causal effect of liquidity. Shadow policies also assume zero latency, fees, rebates, endogenous market impact, and terminal inventory crossed on a copied book. Date-resampling intervals describe variation across held-out dates, whereas policy-choice-change and regret rates are pooled descriptive statistics.

\section{Conclusion}

An aggregate book path does not uniquely determine the FIFO market faced by a passive order. Using synchronized TYO L1 and L2 data, we hold displayed states and reconciled market removals fixed while varying the unobserved placement of within-queue cancellations. For RIC 1301.T, moving from back to front removal raises passive completion by 8.01 percentage points and lowers implementation shortfall by 1.010 bps. A replication on RIC 7911.T preserves the same ordering, with a 7.39-percentage-point completion difference and a 0.384-bps cost difference. The aggressive benchmark remains invariant for both instruments. Policy choices change less frequently but at greater cost for RIC 1301.T, and more frequently but at lower cost for RIC 7911.T. These findings support the paper's claim across two distinct securities: compatible order-level histories can reproduce the same observed aggregate path while implying different passive execution opportunities and, in some episodes, different policy choices. Within the declared compiler class, persistence across matched reconstructions indicates limited dependence on the cancellation rule, whereas a change in policy selection reveals sensitivity to that rule. 

\section{Acknowledgments}
The authors thank London Stock Exchange Group (LSEG) for facilitating access to the proprietary LSEG Tick History for TYO Tokyo Exchange venue historical data used in this study and for approving its use for this research.

This work was completed while Riya Danait was an intern at NVIDIA. Riya Danait's research is supported by \grantsponsor{qrt}{Qube Research and Technologies (QRT)}{qube-rt.com} through the \grantsponsor{epsrc}{Engineering and Physical Sciences Research Council (EPSRC)}{https://www.ukri.org/councils/epsrc/} Centre for Doctoral Training in Mathematics of Random Systems: Analysis, Modelling and Simulation (EPSRC Grant EP/S023925/1). 






\section*{Ethics and Privacy Statement}

We use proprietary TYO L1/L2 market data; access, processing, and reporting follow the data provider's restrictions, and only aggregate results are reported. Price-taking simulation results could be mistaken for guaranteed live-trading performance and used without accounting for market impact or queue model uncertainty. The proposed sensitivity analysis is intended to reduce that risk by identifying execution policies whose apparent advantages depend on unobservable FIFO assumptions before those policies are deployed.

\bibliographystyle{ACM-Reference-Format}
\bibliography{references}

\end{document}